\documentclass[preprint,12pt,3p]{elsarticle}
\usepackage{setspace}
\makeatletter
\def\ps@pprintTitle{%
 \let\@oddhead\@empty
 \let\@evenhead\@empty
 \def\@oddfoot{\footnotesize{\textit{Preprint arXiv}}}%
 \let\@evenfoot\@oddfoot}
\makeatother

\usepackage{amssymb}
\usepackage{booktabs}
\usepackage{multirow}
\usepackage{hyperref}
\hypersetup{
    colorlinks=true,
    linkcolor=blue,
    filecolor=magenta,      
    urlcolor=cyan,
}
 
\usepackage{cmap}				
\usepackage{lmodern}			
\usepackage[T1]{fontenc}		
\usepackage[utf8]{inputenc}		
\usepackage{lastpage}			
\usepackage{indentfirst}		
\usepackage{color}				
\usepackage{graphicx}			
\usepackage{lineno,hyperref}
\usepackage{multirow}
\usepackage{changepage}
\usepackage{longtable}
\usepackage{mathtools}
\usepackage[super,negative]{nth}
\usepackage[caption=false,font=footnotesize]{subfig}
\usepackage{pdflscape}
\usepackage{float}
\usepackage{ltxtable}
\usepackage{rotating}
\usepackage{amsthm,amsmath,bm}

\newtheorem{theorem}{Theorem}

\begin{document}

\begin{frontmatter}

\title{Modelling wind speed with a univariate probability distribution depending on two baseline functions}

\author[label1]{F\'abio V. J. Silveira\corref{cor1}}
\cortext[cor1]{Corresponding author}
\ead{fabio.silveira@ifpb.edu.br}

\author[label2]{Frank Gomes-Silva}

\author[label3]{C\'icero C. R. Brito}

\author[label2]{Jader S. Jale}

\author[label2]{Felipe R. S. Gusm\~ao}

\author[label4]{S\'ilvio F. A. Xavier-J\'unior}

\author[label2]{Jo\~ao S. Rocha}

\address[label1]{Federal Institute of Education, Science and Technology of Para\'iba, Jo\~ao Pessoa, PB, Brazil}

\address[label2]{Department of Statistics and Informatics, Federal Rural University of Pernambuco, Recife, PE, Brazil}

\address[label3]{Federal Institute of Education, Science and Technology of Pernambuco, Recife, PE, Brazil}

\address[label4]{Department of Statistics, Para\'iba State University, Campina Grande, PB, Brazil}

\begin{abstract}
Characterizing the wind speed distribution properly is essential for the satisfactory production of potential energy in wind farms, being the mixture models usually employed in the description of such data. However, some mixture models commonly have the undesirable property of non-identifiability. In this work, we present an alternative distribution which is able to fit the wind speed data adequately. The new model, called Normal-Weibull-Weibull, is identifiable and its cumulative distribution function is written as a composition of two baseline functions. We discuss structural properties of the class that generates the proposed model, such as the linear representation of the probability density function, moments and moment generating function. We perform a Monte Carlo simulation study to investigate the behavior of the maximum likelihood estimates of the parameters. Finally, we present applications of the new distribution for modelling wind speed data measured in five different cities of the Northeastern Region of Brazil.
\end{abstract}


\end{frontmatter}


\section{Introduction and proposed class}

The concern about the emission of greenhouse gases and environmental contamination from conventional energy generation procedures like coal and oil power plants encourages research on alternative resources. A smaller impact on the environment is an advantage of cleaner and sustainable energy production techniques, such as solar, geothermal and wind, over the combustion of fossil fuels. 

The installed capacity of wind power in Brazil increased from 29 MW in 2005 to roughly 16,000 MW (9\% of the total capacity of electricity generation) in 2020 \citep{Araujo2020}. The suitable choice of the wind turbine must match with the wind behavior at the site of installation. \citet{Perkin2015} mention that inadequate turbine selection results in a financially sub-optimal investment. Thus, setting the probability distribution appropriately to model the wind speed is essential. \citet{Eltamaly2013} used the two-parameter Weibull distribution in a new computer program to perform the calculations required to precisely design the wind energy system and to seek the compatibility between sites and turbines.  

\citet{Ilhan2012} declare that despite the wide acceptance of Weibull distribution \citep{Pishgar2015,Safari2011,Weisser2003,Kollu2012}, it may sometimes be poor to model all wind speed data available in nature. Hereupon, they propose using two possible models, namely, the skewed generalized error distribution \citep{Bali2008} and the skewed $t$ distribution \citep{Hansen1994}. Some other distributions used for wind speed and power modelling are Rayleigh \citep{Pishgar2015}, gamma \citep{Morgan2011}, normal \citep{Safari2011}, generalized extreme value \citep{Kollu2012} and Birnbaum-Saunders \citep{Mohammadi2017}. Additionally, applications of nonparametric methods to wind speed modelling are also found in the literature \citep{Qin2011,Bo2016,Han2019}.

Oftentimes, one requires more flexibility from the probability density function (pdf), as in case of bimodality \citep{Jaramillo2004} or calm winds regime \citep{Chang2011}. In general, finite mixture models are more flexible than the typical single ones. \citet{Akdag2010} compared the usual biparametric Weibull and the two-component mixture Weibull distribution in a study focused on wind regimes presenting nearly zero percentage of null speeds; they concluded that the mixture is more suitable to describe such wind conditions. \citet{Carta2007} used three different methods to estimate the parameters of the two-component mixture Weibull, namely, the method of moments, maximum likelihood and least squares; they verified that there is no significant difference among them.       

The mixture density can be written as:
\begin{align} \label{eq:nggmixture}
 f(x;\bm{\psi})=\sum_{i=1}^d w_i f_i(x;\bm{\theta}_i) 
\end{align} 
where the vector $\bm{\psi} = (w_1,\ldots,w_{d-1},\bm{\eta}^{\top})^{\top}$ contains the unknown parameters of the mixture model and the vector $\bm{\eta}$ contains all the distinct parameters in $\bm{\theta}_1, \ldots,\bm{\theta}_d$. Since $w_1,\ldots,w_d$ are positive and sum up to one, the presence of $w_d$ in $\bm{\psi}$ is unnecessary. 

The general definition of identifiability states that a family of densities $\{ f(x;\bm{\psi}): \bm{\psi} \in \bm{\Psi}\}$ is identifiable if:
\begin{align} \label{eq:nggidentif}
 f(x;\bm{\psi}) = f(x;\bm{\psi}^{\star}) \Leftrightarrow \bm{\psi} = \bm{\psi}^{\star}\,.
\end{align} 
It is not seldom that ($\ref{eq:nggidentif}$) fails when two or more component densities in ($\ref{eq:nggmixture}$) belong to the same parametric family. Such is the case of the mixture of normal densities. Consider $d=2$, $f_1$ and $f_2$ are normal densities, $w_1=0.5$ in ($\ref{eq:nggmixture}$) and let $\bm{\eta}_1=(\mu_1,\sigma_1,\mu_2,\sigma_2)^{\top}$ and $\bm{\eta}_2=(\mu_2,\sigma_2,\mu_1,\sigma_1)^{\top}$, where $\mu_1 \neq \mu_2$, $\sigma_1 \neq \sigma_2$. We have $\bm{\psi} = (w_1,\bm{\eta}_1^{\top})^{\top} \neq (w_1,\bm{\eta}_2^{\top})^{\top} = \bm{\psi}^{\star} \Rightarrow f(x;\bm{\psi}) = f(x;\bm{\psi}^{\star})$. That is, ($\ref{eq:nggmixture}$) may be invariant under certain permutations of the elements in the parametric vector. \citet{mclachlan} mention an alternative definition of identifiability for mixture models, such that the mixture of $d$ normal densities would be identifiable under specific conditions. Nonetheless, they remark that it does not overcome the complications due to the interchanging of component labels.

Models like the mixture of normal or Weibull densities are quite flexible tools, although the parametric estimation is only credible if the distribution under study is identifiable. We present in this paper a class, whose submodels may be feasible alternatives to mixtures of two components from the same parametric family. The class is derived using the method of generating classes of probability distributions of \citet{Brito2019}. Its cumulative distribution function (cdf) is formulated as a composition of two baselines and under certain conditions, it satisfies (\ref{eq:nggidentif}), even if both baselines belong to the same parametric family.    

The cdf of the general class is given by:
\begin{equation} \label{eq:geradorBritongg}
 F(x) = \zeta (x) \sum_{j=1}^n \int_{L_j (x)}^{U_j (x)} \mathrm{d}H (t) - \nu (x) \sum_{j=1}^n \int_{M_j (x)}^{V_j (x)} \mathrm{d}H (t) 
\end{equation} where $H$ is a cdf, $n \in \mathbb{N}$, $\zeta, \nu : \mathbb{R} \mapsto \mathbb{R}$ and 
$L_j, U_j, M_j, V_j: \mathbb{R} \mapsto \mathbb{R} \cup \{ \pm \infty\}$ are special functions that will be discussed in the next section.

\subsection{The Normal-$(G_1,G_2)$ class and some structural properties}
The method established by~\citet{Brito2019} states that if $H, \zeta, \nu : \mathbb{R} \mapsto \mathbb{R}$ and 
$L_j, U_j, M_j, V_j: \mathbb{R} \mapsto \mathbb{R} \cup \{ \pm \infty\}$ for $j=1,2,3,\ldots,n$ are monotonic and right continuous functions such that: 
\begin{itemize}
 \item[(c1)] $H$ is a cdf and $\zeta$ and $\nu$ are non-negative;
\item[(c2)] $\zeta (x)$, $U_j (x)$ and $M_j (x)$ are non-decreasing and $\nu(x)$, $V_j(x)$, $L_j(x)$ are non-increasing $\forall j=1,2,3,\ldots,n$;
\item[(c3)] If $\displaystyle \lim_{x \to -\infty} \zeta (x) \neq \lim_{x \to -\infty} \nu(x)$, then 
$\displaystyle \lim_{x \to -\infty} \zeta (x)=0$; \textbf{or} \\ $ \displaystyle \lim_{x \to -\infty} U_j (x) = \lim_{x \to -\infty} L_j(x)\, \forall j=1,2,3,\ldots,n$, and $ \displaystyle \lim_{x \to -\infty} \nu (x)= 0$; \textbf{or} \\$\displaystyle \lim_{x \to -\infty} M_j (x) = \lim_{x \to -\infty} V_j(x)\, \forall j=1,2,3,\ldots,n$;
\item[(c4)] If $\displaystyle \lim_{x \to -\infty} \zeta (x) = \lim_{x \to -\infty} \nu(x) \neq 0$, then
$\displaystyle \lim_{x \to -\infty} U_j (x) = \lim_{x \to -\infty} V_j(x)$ and $\displaystyle \lim_{x \to -\infty} M_j (x) = \lim_{x \to -\infty} L_j(x)\, \forall j=1,2,3,\ldots,n$;
\item[(c5)] $\displaystyle \lim_{x \to -\infty} L_j (x) \leq \lim_{x \to -\infty} U_j(x)$ and if $\displaystyle \lim_{x \to -\infty} \nu(x) \neq 0$, then $\displaystyle \lim_{x \to +\infty} M_j (x) \leq \lim_{x \to +\infty} V_j(x)\, \forall j=1,2,3,\ldots,n$;
\item[(c6)] $\displaystyle \lim_{x \to +\infty} U_n (x) \geq \sup \{ x \in \mathbb{R}: H(x)<1 \}$ and $\displaystyle \lim_{x \to +\infty} L_1 (x) \leq \inf \{ x \in \mathbb{R}: H(x)>0 \}$;
\item[(c7)] $\displaystyle \lim_{x \to +\infty} \zeta (x)=1$;
\item[(c8)] $\displaystyle \lim_{x \to +\infty} \nu (x)=0$ or $\displaystyle \lim_{x \to +\infty} M_j (x) = \lim_{x \to +\infty} V_j(x)\, \forall j=1,2,3,\ldots,n$ and $n \geq 1$; 
\item[(c9)] $\displaystyle \lim_{x \to +\infty} U_j (x) = \lim_{x \to +\infty} L_{j+1}(x)\, \forall j=1,2,3,\ldots,n-1$ and $n \geq 2$; 
\item[(c10)] $H$ is a cdf without points of discontinuity or all functions $L_j(x)$ and $V_j(x)$ are constant at the right of the vicinity of points whose image are points of discontinuity of $H$, being also continuous in that points. Moreover, $H$ does not have any point of discontinuity in the set $\displaystyle \left\{ \displaystyle \lim_{x \to \pm \infty} L_j (x), \lim_{x \to \pm \infty} U_j (x), \lim_{x \to \pm \infty} M_j (x), \lim_{x \to \pm \infty} V_j (x)\right\}$ for some $j=1,2,3,\ldots,n$;
\end{itemize} then Equation~(\ref{eq:geradorBritongg}) is a cdf.

Let $n=1$, $H(t) =\mathrm{\Phi}(t)$, namely, the standard normal cdf, $\zeta(x)=1$, $\nu(x)=0$, $U_1 (x) = G_1(x)/[1-G_1(x)]$ and $L_1(x) = \log[1-G_2(x)]$, where $G_1(x)$ and $G_2(x)$ are cdfs. The function in Equation~(\ref{eq:geradorBritongg}) turns into:
\begin{equation} \label{eq:cdfngg}
 F_{G_1,G_2}(x) = \int^{\frac{G_1(x)}{1-G_1(x)}}_{\log[1-G_2(x)]} \mathrm{d}\mathrm{\Phi}(t).
\end{equation} We took $U_1$ and $-L_1$ from the table of differentiable and monotonically non-decreasing functions presented in the well-known paper of \citet{Alzaatreh2013}, whose method was used to create generalized distributions of the T-X family. We have intentionally picked the two simplest functions from the cited table; alternative (and more complicated) choices for $U_1$ and $L_1$ would naturally give rise to different classes.
Defining $M_1(x)$ and $V_1(x)$ is not relevant, since $\nu(x)=0$. Also, for obvious reasons, there is no need to verify (c4), (c5) and (c9). The conditions (c1), (c7), (c8) and (c10) are straightforward. As $U_1(x)$ and $\zeta(x)$ are non-decreasing and $L_1(x)$ is non-increasing, (c2) is true. 
It is easy to verify that $\displaystyle \lim_{x\rightarrow-\infty} U_1(x) = 0 = \lim_{x\rightarrow-\infty} L_1(x)$; and since $\displaystyle \lim_{x\rightarrow-\infty} \nu(x) = 0$, (c3) is satisfied. The condition (c6) is also true because $\displaystyle \lim_{x \to +\infty} U_1(x) = +\infty = \sup \{ x \in \mathbb{R}: \mathrm{\Phi}(x)<1 \}$ and $\displaystyle \lim_{x \to +\infty} L_1(x) = -\infty = \inf \{ x \in \mathbb{R}: \mathrm{\Phi}(x)>0\}$.

Thereby, in agreement with the method exposed above, Equation~(\ref{eq:cdfngg}) is a cdf. As already mentioned, it can be viewed as a composite function of two baselines. Henceforth, let it be denoted by Normal-$(G_1,G_2)$ class of probability distributions.

Since $\phi(t) = \frac{1}{\sqrt{2 \pi}} e^{-t^2 /2}$, and $\mathrm{\Phi}(x) = \int_{-\infty}^x \phi(t) \mathrm{d}t$, one can write Equation~(\ref{eq:cdfngg}) as follows:
\begin{equation} \label{eq:cdfnggdiff}
 F_{G_1,G_2}(x) = \mathrm{\Phi} \left( \frac{G_1(x)}{1-G_1(x)} \right) - \mathrm{\Phi} \left(\log[1-G_2(x)] \right) \,.
\end{equation} 
In case of continuous $G_1(x)$ and $G_2(x)$, one can take the derivative of Equation~(\ref{eq:cdfnggdiff}) with respect to $x$ to obtain the following pdf:
\begin{equation}
 f_{G_1,G_2}(x) = \phi \left(\frac{G_1(x)}{1-G_1(x)} \right) \frac{g_1(x)}{[1-G_1(x)]^2} + \phi \left( \log[1-G_2(x)] \right) \frac{g_2(x)}{1-G_2(x)}, \label{eq:pdfngg}
\end{equation} where $g_i(x)$ is the pdf of the random variable whose cdf is $G_i(x)$, for $i \in \{1,2\}$.

At this point, we need to define properly the support of the distributions that emerge from the new class. Submodels of classes that may be written as a composite function of one single baseline usually have the same support of the baseline. However, characterizing the support of a submodel from~(\ref{eq:cdfngg}) is not so straightforward, especially if the two baselines have different supports. As previously mentioned, given that $U_1 (G_1(x),G_2(x)) = G_1(x)/[1-G_1(x)]$, $L_1(G_1(x),G_2(x)) = \log[1-G_2(x)]$ and $S_H=\mathbb{R}$, namely, the support of $H(t)$ is $\mathbb{R}$, we have that:
\begin{itemize}
 \item[(a)] $S_H$ is a convex set;
 \item[(b)] $U_1(1,1) = U_1(G_1(+\infty),G_2(+\infty)) = +\infty = \sup \{ x \in \mathbb{R}: \mathrm{\Phi}(x)<1\}$, $L_1(1,1) = L_1(G_1(+\infty),G_2(+\infty)) = -\infty = \inf \{ x \in \mathbb{R}: \mathrm{\Phi}(x)>0\}$, $U_1 (G_1(x),G_2(x))$ and $L_1 (G_1(x),G_2(x))$ are monotonic functions.   
\end{itemize}
According to the Theorem (T4) in~\citep{Brito2019}, (a) and (b) entail that the support of a distribution from~(\ref{eq:cdfngg}) is the union of the supports of $G_1$ and $G_2$.

In the following lines, we demonstrate that, under specific conditions, the distributions generated by (\ref{eq:cdfnggdiff}) enjoy the attractive property of identifiability. It is important because it assures the uniqueness of the estimates of the parameters.

\begin{theorem} \label{teo:ngg}

 Let $G_1(x|\bm{\theta}_1)$ and $G_2(x|\bm{\theta}_2)$ be the baseline cdfs of the normal-$(G_1,G_2)$ cdf $F_{G_1,G_2}(x|\bm{\theta})$ (\ref{eq:cdfnggdiff}), $\bm{\theta}_1 = (\theta_{1},\ldots,\theta_{r}) \in \bm{\Theta}_1$, $\bm{\theta}_2 = (\theta_{r+1},\ldots,\theta_{r+m}) \in \bm{\Theta}_2$ and \\$\bm{\theta} = (\theta_{1},\ldots,\theta_{r},\theta_{r+1},\ldots,\theta_{r+m}) \in \bm{\Theta}$, where $\bm{\Theta}_1$, $\bm{\Theta}_2$ and $\bm{\Theta}$ are the parametric spaces associated with $G_1$, $G_2$ and $F_{G_1,G_2}$ respectively. If $G_1$ and $G_2$ are identifiable, then $F_{G_1,G_2}$ is identifiable.
\end{theorem}

\begin{proof} 
Assume that $\mathrm{\Phi} \left( \frac{G_1 (x|\bm{\theta}_1)}{1-G_1 (x|\bm{\theta}_1)} \right) = \mathrm{\Phi} \left( \frac{G_1 (x|\bm{\theta}_1^{\star})}{1-G_1 (x|\bm{\theta}_1^{\star})} \right)$, where $\{ \bm{\theta}_1 , \bm{\theta}_1^{\star}\} \subset \bm{\Theta}_1$ and $\bm{\theta}_1 \neq \bm{\theta}_1^{\star}$. Since $\mathrm{\Phi}$ is injective, $ \frac{G_1 (x|\bm{\theta}_1)}{1-G_1 (x|\bm{\theta}_1)} = \frac{G_1 (x|\bm{\theta}_1^{\star})}{1-G_1 (x|\bm{\theta}_1^{\star})} \Rightarrow G_1 (x|\bm{\theta}_1) = G_1 (x|\bm{\theta}_1^{\star})$; it is a contradiction, because it denies the identifiability of $G_1$. Therefore, if $\bm{\theta}_1 \neq \bm{\theta}_1^{\star}$ then $\mathrm{\Phi} \left( \frac{G_1 (x|\bm{\theta}_1)}{1-G_1 (x|\bm{\theta}_1)} \right) \neq \mathrm{\Phi} \left( \frac{G_1 (x|\bm{\theta}_1^{\star})}{1-G_1 (x|\bm{\theta}_1^{\star})} \right)$. Analogously, it is easy to verify that for $\{ \bm{\theta}_2 , \bm{\theta}_2^{\star}\} \subset \bm{\Theta}_2$, if $\bm{\theta}_2 \neq \bm{\theta}_2^{\star}$ then $\mathrm{\Phi} \left( \log[1-G_2(x|\bm{\theta}_2)] \right) \neq \mathrm{\Phi} \left( \log[1-G_2(x|\bm{\theta}_2^{\star})] \right)$.

Now consider $\{ \bm{\theta} , \bm{\theta}^{\star}\} \subset \bm{\Theta}$ such that $\bm{\theta} \neq \bm{\theta}^{\star}$ and assume that $F_{G_1,G_2}(x|\bm{\theta}) = F_{G_1,G_2}(x|\bm{\theta}^{\star})$. If $\bm{\theta}_1 = \bm{\theta}_1^{\star}$ and $\bm{\theta}_2 \neq \bm{\theta}_2^{\star}$, then we can infer from (\ref{eq:cdfnggdiff}) that $G_2(x|\bm{\theta}_2) = G_2(x|\bm{\theta}_2^{\star})$, namely, an absurd. Likewise, if $\bm{\theta}_1 \neq \bm{\theta}_1^{\star}$ and $\bm{\theta}_2 = \bm{\theta}_2^{\star}$, we get to similar contradiction. If $\bm{\theta}_1 \neq \bm{\theta}_1^{\star}$ and $\bm{\theta}_2 \neq \bm{\theta}_2^{\star}$, then the assumption fails since $F_{G_1,G_2}(x|\bm{\theta}) \neq F_{G_1,G_2}(x|\bm{\theta}^{\star})$ for almost all values of $x$ within the support. Therefore, $F_{G_1,G_2}$ is identifiable. 

\end{proof}

\subsection{Series representation}
The normal cdf can be written in terms of the error function erf as follows:
\begin{equation} \label{eq:cdferfngg}
 \mathrm{\Phi}(z) = \frac{1}{2} \left[ 1 + \mathrm{erf} \left( \frac{z}{\sqrt{2}} \right) \right]\;, 
\end{equation} 
where $\mathrm{erf}(z) = \frac{2}{\sqrt{\pi}} \int_0^z e^{-t^2} \mathrm{d}t$. Since $\mathrm{erf}(z/ \sqrt{2} )$ may be linearly represented by:
\begin{align}
 \mathrm{erf} \left( \frac{z}{\sqrt{2}} \right) & = \frac{2}{\sqrt{\pi}} \sum_{n=0}^{\infty} \frac{(-1)^n \cdot (z/\sqrt{2})^{2n+1}}{n! (2n+1)} \nonumber \\
	& = \sqrt{\frac{2}{\pi}} \cdot \sum_{n=0}^{\infty} \left(-\frac{1}{2} \right)^n \frac{z^{2n+1}}{n! (2n+1)}\;, \label{eq:erf2}
\end{align} replacing Equation~(\ref{eq:erf2}) in Equation~(\ref{eq:cdferfngg}), we get to: 
\begin{equation} 
 \mathrm{\Phi}(z) = \frac{1}{2} + \frac{1}{\sqrt{2 \pi}} \sum_{n=0}^{\infty} \left( -\frac{1}{2} \right)^n \frac{z^{2n+1}}{n! (2n+1)} \;. \label{eq:cdf2} 
\end{equation}
Now using the result of Equation~(\ref{eq:cdf2}) in Equation~(\ref{eq:cdfnggdiff}), we have:
\begin{align}
 F_{G_1,G_2}(x) = \sum_{n=0}^{\infty} \frac{(-1/2)^n}{n!(2n+1)\sqrt{2\pi}} \left[ \underbrace{\left( \frac{G_1(x)}{1-G_1(x)} \right)^{2n+1}}_{\mathrm{A1}} - \underbrace{\left( \log[1-G_2(x)] \right)^{2n+1}}_{\mathrm{A2}} \right]. \label{eq:ngg001}
 \end{align}

A well-known result on power series raised to a positive integer $N$ states that:
\begin{equation}
 \left( \sum_{k=0}^{\infty} a_k y^k \right)^N = \sum_{k=0}^{\infty} c_k y^k\;, \label{eq:seriesraised}
\end{equation} where $c_0 = a_0^N$, $c_k = \frac{1}{k a_0} \sum_{s=1}^k (sN - k + s) a_s c_{k-s}$ for $k \geq 1$ and $N \in \mathbb{N}$.
Setting $N = 2n+1$, $y=G_1(x)$ and $a_k = 1$ for all $k \geq 0$, we can use the result in Equation~(\ref{eq:seriesraised}) to rewrite A1 in Equation~(\ref{eq:ngg001}):
\begin{align} 
\mathrm{A1} & = G_1(x)^{2n+1} \left( \frac{1}{1-G_1(x)} \right)^{2n+1} = G_1(x)^{2n+1} \left( \sum_{k=0}^{\infty} G_1(x)^k \right)^{2n+1} \nonumber \\
  & = G_1(x)^{2n+1} \sum_{k=0}^{\infty} c_{1,k} \cdot G_1(x)^k = \sum_{k=0}^{\infty} c_{1,k} \cdot G_1(x)^{k+2n+1}\;, \label{eq:a1ngg}
\end{align} 
such that $c_{1,0} = 1$ and $c_{1,k} = \frac{1}{k} \sum_{s=1}^k ( 2s[n+1]-k ) c_{1,k-s}$ for $k \geq 1$.
Equation~(\ref{eq:seriesraised}) also allows us to rewrite A2 in Equation~\ref{eq:ngg001} as follows:
\begin{align}
 \mathrm{A2} &= \left(- \sum_{m=1}^{\infty} \frac{G_2(x)^m}{m} \right)^{2n+1} = -\left( \sum_{k=0}^{\infty} \frac{G_2(x)^{k+1}}{k+1} \right)^{2n+1} \nonumber \\
 & = - G_2(x)^{2n+1} \left( \sum_{k=0}^{\infty} \frac{G_2(x)^{k}}{k+1} \right)^{2n+1} = -G_2(x)^{2n+1} \sum_{k=0}^{\infty} c_{2,k} \cdot G_2(x)^k \nonumber \\
 & = - \sum_{k=0}^{\infty} c_{2,k} \cdot G_2(x)^{k+2n+1} \label{eq:a2ngg}
\end{align} where $c_{2,0}=1$ and $c_{2,k} = \frac{1}{k} \sum_{s=1}^{k} \frac{2s(n+1)-k}{s+1} c_{2,k-s}$ for $k \geq 1$. Now inserting (\ref{eq:a1ngg}) and (\ref{eq:a2ngg}) in (\ref{eq:ngg001}), we have:
\begin{align}
 F_{G_1,G_2}(x) = \sum_{i=1}^2 \sum_{n,k=0}^{\infty} c_{i,n,k} \cdot G_i(x)^{k+2n+1} \label{eq:ngg003}
\end{align} where $c_{1,n,k} = \frac{(-1/2)^n}{n!(2n+1)\sqrt{2\pi}} c_{1,k}$ and $c_{2,n,k} = \frac{(-1/2)^n}{n!(2n+1)\sqrt{2\pi}} c_{2,k}$. Using Fubini's theorem on differentiation we can write the derivative of (\ref{eq:ngg003}) as follows:
\begin{equation} \label{eq:ngg004}
 f_{G_1,G_2}(x) = \sum_{i=1}^2 \sum_{n,k=0}^{\infty} c_{i,n,k} \cdot g_{i,k+2n+1}(x) 
\end{equation} where $g_{i,k+2n+1}(x)=(k+2n+1) g_i(x) G_i(x)^{k+2n}$ is the pdf of a random variable from the exponentiated family \citep{Mudholkar1993}. Thus, we can say that~(\ref{eq:ngg004}) is the Normal-$(G_1,G_2)$ pdf~(\ref{eq:pdfngg}) expressed as a linear combination of pdfs of exponentiated distributions.

\subsection{Raw moments, incomplete moments and moment generating function}

Given that $X$ is a random variable following a distribution from the normal-$(G_1,G_2)$ class, we can use (\ref{eq:ngg004}) to write the $r$-th raw moment of $X$ as follows:
\begin{align}
 E(X^r) & = \sum_{i=1}^2 \sum_{n,k=0}^{\infty} c_{i,n,k} \int_{-\infty}^{\infty} x^r g_{i,k+2n+1}(x) \mathrm{d}x \label{eq:nggmoment0} \\ 
 & = \sum_{i=1}^2 \sum_{n,k=0}^{\infty} c_{i,n,k} E(Y_{i,k+2n+1}^r) \label{eq:nggmoment}
\end{align} where $Y_{i,k+2n+1}$ follows the exponentiated distribution whose pdf is $g_{i,k+2n+1}$. 

Let $Q_i$ be the quantile function of the baseline $G_i$. Replacing $x$ in (\ref{eq:nggmoment0}) by $Q_i \left( v^{{1}/{k+2n+1}} \right)$ we can also represent (\ref{eq:nggmoment}) as:
\begin{equation*}
 E(X^r) = \sum_{i=1}^2 \sum_{n,k=0}^{\infty} c_{i,n,k} \int_0^1 \left[ Q_i \left( v^{{1}/{k+2n+1}} \right) \right]^r \mathrm{d}v\,.
\end{equation*}

Similarly, one can write the $r$-th incomplete moment of $X$ as follows: 
\begin{align}
 m_r(z) & = \int_{-\infty}^z x^r f_{G_1,G_2}(x) \mathrm{d}x = \sum_{i=1}^2 \sum_{n,k=0}^{\infty} c_{i,n,k} m_r^{\star}(z) \nonumber \\
 & = \sum_{i=1}^2 \sum_{n,k=0}^{\infty} c_{i,n,k} \int_0^{\left[ G_i (z) \right]^{k+2n+1}} \left[ Q_i \left( v^{{1}/{k+2n+1}} \right) \right]^r \mathrm{d}v \nonumber
\end{align} where $m_r^{\star}(z)$ is the $r$-th incomplete moment of $Y_{i,k+2n+1}$ mentioned above. 
 
The moment generating function (mgf) of $X$ is denoted by $M_X(t)=E \left( e^{tX} \right)$. It can be determined from (\ref{eq:ngg004}) as: 
\begin{align}
 M_X(t) & = \sum_{i=1}^2 \sum_{n,k=0}^{\infty} c_{i,n,k} \int_{-\infty}^{\infty} e^{tx} g_{i,k+2n+1}(x) \mathrm{d}x \nonumber \\
 & = \sum_{i=1}^2 \sum_{n,k=0}^{\infty} c_{i,n,k} M_{Y_{k+2n+1}}(t)\,, \nonumber
\end{align} where $M_{Y_{k+2n+1}}(t)$ is the mgf of $Y_{i,k+2n+1}$. 

Other meaningful quantities, as the characteristic function, the probability-weighted moments, the R\'enyi entropy and the order statistics can be derived likewise by using (\ref{eq:ngg004}).


\subsection{Estimation and Inference} \label{sec:ngginference}

Let $\bm{X}=(x_1,\ldots,x_n)$ be a complete random sample of size $n$ from the random variable $X \sim$ normal-$(G_1,G_2)$. Given that $\bm{\theta}_1 = (\theta_1,\ldots,\theta_r)^{\top}$ is the $r \times 1$ parametric vector associated with $G_1(x)=G_1(x|\bm{\theta}_1)$, $\bm{\theta}_2 = (\theta_{r+1},\ldots,\theta_{r+m})^{\top}$ is the $m \times 1$ parametric vector associated with $G_2(x)=G_2(x|\bm{\theta}_2)$ and $f_{G_1,G_2}(x)=f_{G_1,G_2}(x|\bm{\theta})$ where $\bm{\theta} = (\theta_1,\ldots,\theta_r,\theta_{r+1},\ldots,\theta_{r+m})^{\top}$, we can write the log-likelihood function of $X$ as follows:
\begin{equation*}
 \ell(\bm{\theta}|\bm{X}) = \sum_{i=1}^n \log \left\{ \phi \left(\frac{G_1(x_i)}{1-G_1(x_i)} \right) \frac{g_1(x_i)}{[1-G_1(x_i)]^2} + \phi \left( \log[1-G_2(x_i)] \right) \frac{g_2(x_i)}{1-G_2(x_i)} \right\}\,. 
\end{equation*} The solution of the system of equations $U(\bm{\theta}|\bm{X})=\bm{0}_{r+m}$ provides the maximum likelihood estimates (MLEs) for $\bm{\theta}$, where $\bm{0}_{r+m}$ is an $(r+m)\times 1$ vector of zeros and $U(\bm{\theta}|\bm{X}) = \nabla_{\bm{\theta}} \ell (\bm{\theta}|\bm{X})$ is the score vector. The elements of $U(\bm{\theta}|\bm{X})=(u_j)_{1 \leq j \leq r+m}$ are:
\begin{multline*}
 u_j = \sum_{i=1}^n \frac{1}{f_{G_1,G_2}(x_i)} \phi \left(\frac{G_1(x_i)}{1-G_1(x_i)} \right) \frac{1}{(1-G_1(x_i))^2} \left[ \frac{\partial}{\partial \theta_j} {g_1}(x_i) \right. \\ 
 \left. + \frac{g_1(x_i)}{1-G_1(x_i)} \left( 2- \frac{G_1(x_i)}{[1-G_1(x_i)]^2} \right) \frac{\partial}{\partial \theta_j} G_1(x_i) \right],\; \mathrm{for}\; 1 \leq j \leq r
\end{multline*} 
and
\begin{multline*}
 u_j = \sum_{i=1}^n \frac{1}{f_{G_1,G_2}(x_i)} \frac{\phi \left( \log[1-G_2(x_i)] \right)}{1-G_2(x_i)} \left[ \frac{\partial}{\partial \theta_j} g_2(x_i) \right. \\
 \left. + \left( 1+\log[1-G_2(x_i)] \right) \frac{g_2(x_i)}{1-G_2(x_i)} \frac{\partial}{\partial \theta_j} G_2(x_i) \right],\; \mathrm{for}\; r < j \leq r+m\,.
\end{multline*}  

For testing hypotheses and constructing confidence intervals for $\bm{\theta}$, the information matrix $J(\bm{\theta}|\bm{X})$ is needed. The expectation of $J(\bm{\theta}|\bm{X})$, denoted by $\mathcal{I}_{\bm{\theta}}$, is the expected Fisher information matrix. Given that certain conditions of regularity are fulfilled, the quantity $\sqrt{n}( \widehat{ \bm{\theta}} - \bm{\theta})$ follows approximately a multivariate normal distribution $N_{r+m}(\mathbf{0}_{r+m},\mathcal{I}_{\bm{\theta}}^{-1})$.~\ref{app:a} brings the expression for $J(\bm{\theta}|\bm{X})$.

\section{The proposed model} \label{sec:model}
The Weibull cdf is given by $G_W(x|k,\lambda) = 1-e^{-(x/\lambda)^k}$, for $x \geq 0$, $k >0$ and $\lambda>0$. Replacing $G_1$ and $G_2$ in (\ref{eq:cdfnggdiff}) by $G_W(x|k_1, \lambda_1)$ and $G_W(x|k_2,\lambda_2)$ respectively, we get to the cdf of the Normal-Weibull-Weibull distribution (NWW, for short):
\begin{align}
 F_{NWW}(x|\bm{\theta}) =  \mathrm{\Phi} \left( e^{(x/\lambda_1)^{k_1}}-1 \right) - \mathrm{\Phi} \left( -\left( \frac{x}{\lambda_2} \right)^{k_2} \right)\,, \nonumber
\end{align} where $\bm{\theta}=(k_1,\lambda_1,k_2,\lambda_2)^{\top}$. The corresponding pdf can be obtained using (\ref{eq:pdfngg}):
\begin{equation}
 f_{NWW}(x|\bm{\theta}) = \phi \left( e^{(x/\lambda_1)^{k_1}}-1 \right) \frac{k_1}{\lambda_1} \left( \frac{x}{\lambda_1} \right)^{k_1-1} e^{(x/\lambda_1)^{k_1}}
 + \phi \left( -\left( \frac{x}{\lambda_2} \right)^{k_2} \right) \frac{k_2}{\lambda_2} \left( \frac{x}{\lambda_2} \right)^{k_2-1}.
\end{equation}

Figure~\ref{fig:nwwpdf} displays some plots of the NWW pdf for different values of the parameters. The distribution is able to fit unimodal right-skewed data (top-left chart) and also left-skewed data (top-right chart). Notice the different shapes of the bimodal curves in the remaining charts. For instance, in the bottom-left chart, the vertical distance between the modes and the local minimum in the purple curve is much greater than in the green one. We may also notice that $\lambda_1$ and $\lambda_2$ somehow behave like shape parameters, as in the original Weibull baselines, controlling the shape of the ``bells'' (compare purple and gray curves to see the effect of varying $\lambda_1$; same for blue and red curves concerning $\lambda_2$).     


\begin{figure}[h]
 \centering
 \includegraphics[scale=0.64]{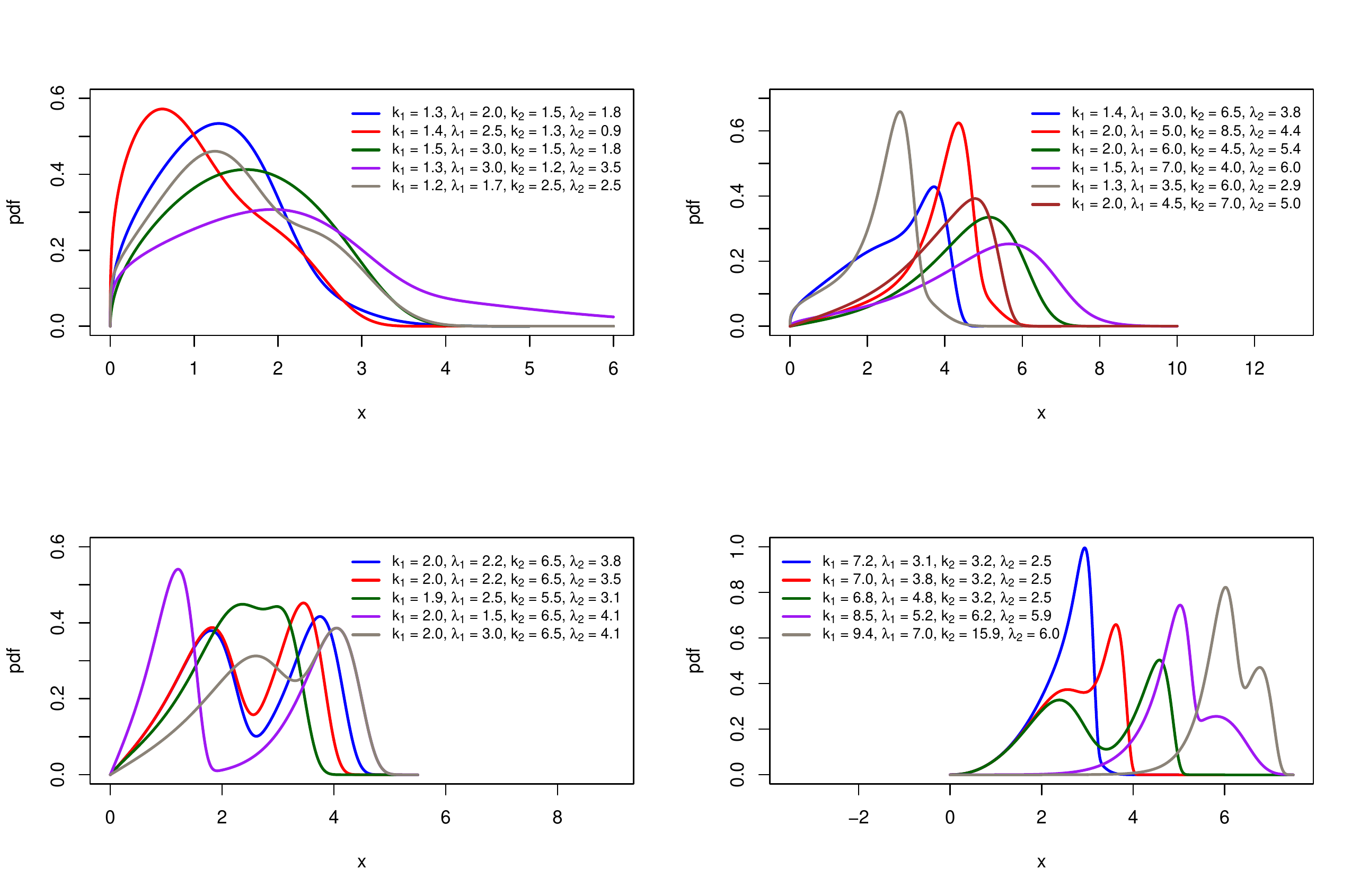}
 \caption{Plots of the Normal-Weibull-Weibull pdf.}
 \label{fig:nwwpdf}
\end{figure}

\newpage

\section{Simulation}

Performing Monte Carlo simulation studies is considerably relevant whenever one wants to test and confirm assumptions on new statistical methods. In this work, we want to investigate the behavior of the estimates of the parameters of the NWW distribution under the method of maximum likelihood. For this purpose, we used the software R version 3.4.4~\citep{r2018}.

Initially, we employed the Von Neumann's acceptance-rejection method~\citep{Neumann} to generate pseudo-random samples from the NWW distribution; this simple method requires only the corresponding pdf $y=f(x)$, a minorant and a majorant for $x$ and a majorant for $y$. The procedure was replicated 10,000 times and at each replication, four different sample sizes were considered, namely, $n=50$, $100$, $200$ and $500$. We examined scenarios with four different values of the parametric vector $\bm{\theta}=(k_1,\lambda_1,k_2,\lambda_2)^{\top}$, which are presented from the second to fifth columns of Tables~\ref{tab:nggbias} and~\ref{tab:nggmse}. 

For each scenario, we calculated the bias and the mean squared error (MSE) as follows:
\begin{equation*}
 {\mathrm{Bias}}_i = \frac{1}{10000} \sum_{j=1}^{10000} \left( {\widehat{\theta}}_{ij} - \theta_i \right)\;, \qquad {\mathrm{MSE}_i} = \frac{1}{10000} \sum_{j=1}^{10000} \left( {\widehat{\theta}}_{ij} - \theta_i \right)^2
\end{equation*} where $\theta_i$ is the $i$-th element of $\bm{\theta}$ and ${\widehat{\theta}}_{ij}$ is the estimate for $\theta_i$ at the $j$-th replication.

The global maximum of the log-likelihood function was found by using the L-BFGS-B algorithm. It is based on the gradient projection and uses a limited memory BFGS matrix to approximate the Hessian of the objective function~\citep{Byrd1995}. 

Besides presenting small values, the desired behavior for both bias and MSE is to decrease inasmuch as the sample size increases. According to Tables~\ref{tab:nggbias} and~\ref{tab:nggmse}, the values of bias and MSE for all the estimated parameters are small and the greater the sample size, the smaller the values. Thus, the results presented in this section indicate that the MLEs of the parameters of the NWW distribution are well-behaved when calculated using the L-BFGS-B algorithm.   

\begin{table}[h!]
\caption{Bias of the estimates under the maximum likelihood method for the NWW model.}
\label{tab:nggbias} 
\centering
\begin{tabular}{lcccccccc}
\hline\noalign{\smallskip}
&\multicolumn{4}{c}{Actual value} & \multicolumn{4}{c}{Bias} \\
\cmidrule(lr){2-5} \cmidrule(lr){6-9}
$n$ & $k_1$ & $\lambda_1$ & $k_2$ & $\lambda_2$ & $\widehat k_1$ & $\widehat \lambda_1$ & $\widehat k_2$ & $\widehat \lambda_2$ \\
\noalign{\smallskip}\hline\noalign{\smallskip}
50	&	1.3	&	2	&	1.5	&	1.8	&	0.54169	&	0.30227	&	0.42743	&	0.29174 \\
	&	3	&	1.5	&	2.8	&	2.5	&	0.55544	&	0.23582	&	0.63534	&	0.2146 \\
	&	2	&	2.2	&	6.5	&	4.1	&	0.30213	&	0.1419	&	1.42262	&	0.12017 \\
	&	1.4	&	1.6	&	4.8	&	5.1	&	0.18208	&	0.10899	&	0.86019	&	0.13781 \\
	\cmidrule{1-9}
100	&	1.3	&	2	&	1.5	&	1.8	&	0.34967	&	0.24391	&	0.31488	&	0.22752 \\
	&	3	&	1.5	&	2.8	&	2.5	&	0.43866	&	0.1199	&	0.41435	&	0.13804 \\
	&	2	&	2.2	&	6.5	&	4.1	&	0.24521	&	0.11322	&	0.89279	&	0.08923 \\
	&	1.4	&	1.6	&	4.8	&	5.1	&	0.13014	&	0.07167	&	0.54369	&	0.09648 \\
	\cmidrule{1-9}
200	&	1.3	&	2	&	1.5	&	1.8	&	0.22345	&	0.18795	&	0.24873	&	0.17187 \\
	&	3	&	1.5	&	2.8	&	2.5	&	0.28333	&	0.04346	&	0.26341	&	0.07261 \\
	&	2	&	2.2	&	6.5	&	4.1	&	0.161	&	0.06349	&	0.57998	&	0.05034 \\
	&	1.4	&	1.6	&	4.8	&	5.1	&	0.09163	&	0.04958	&	0.37906	&	0.069 \\
    \cmidrule{1-9}
500	&	1.3	&	2	&	1.5	&	1.8	&	0.13429	&	0.11752	&	0.17774	&	0.11164 \\
	&	3	&	1.5	&	2.8	&	2.5	&	0.16708	&	0.01806	&	0.17836	&	0.04992 \\
	&	2	&	2.2	&	6.5	&	4.1	&	0.09747	&	0.03622	&	0.34572	&	0.0277 \\
	&	1.4	&	1.6	&	4.8	&	5.1	&	0.06333	&	0.03001	&	0.22205	&	0.04259 \\
\noalign{\smallskip}\hline
\end{tabular}
\end{table}     

\begin{table}[h!]
\caption{MSE of the estimates under the maximum likelihood method for the NWW model.}
\label{tab:nggmse} 
\centering
\begin{tabular}{lcccccccc}
\hline\noalign{\smallskip}
&\multicolumn{4}{c}{Actual value} & \multicolumn{4}{c}{MSE} \\
\cmidrule(lr){2-5} \cmidrule(lr){6-9}
$n$ & $k_1$ & $\lambda_1$ & $k_2$ & $\lambda_2$ & $\widehat k_1$ & $\widehat \lambda_1$ & $\widehat k_2$ & $\widehat \lambda_2$ \\
\noalign{\smallskip}\hline\noalign{\smallskip}
50 &	1.3	&	2	&	1.5	&	1.8	&	0.5477	&	0.13969	&	0.37049	&	0.13065	\\
 &	3	&	1.5	&	2.8	&	2.5	&	0.43676	&	0.2146	&	0.88316	&	0.11482	\\
 &	2	&	2.2	&	6.5	&	4.1	&	0.1976	&	0.07682	&	3.60838	&	0.05433	\\
 &	1.4	&	1.6	&	4.8	&	5.1	&	0.05343	&	0.02546	&	1.32698	&	0.03216	\\
 \cmidrule{1-9}
100 &	1.3	&	2	&	1.5	&	1.8	&	0.23088	&	0.09363	&	0.17553	&	0.0809	\\
 &	3	&	1.5	&	2.8	&	2.5	&	0.27616	&	0.09641	&	0.30237	&	0.0571	\\
 &	2	&	2.2	&	6.5	&	4.1	&	0.16973	&	0.07868	&	1.39857	&	0.06496	\\
 &	1.4	&	1.6	&	4.8	&	5.1	&	0.02719	&	0.00888	&	0.49768	&	0.01486	\\
 \cmidrule{1-9}
200 &	1.3	&	2	&	1.5	&	1.8	&	0.0914	&	0.0588	&	0.10077	&	0.04583	\\
 &	3	&	1.5	&	2.8	&	2.5	&	0.12419	&	0.01817	&	0.11285	&	0.01472	\\
 &	2	&	2.2	&	6.5	&	4.1	&	0.05914	&	0.02211	&	0.57197	&	0.01905	\\
 &	1.4	&	1.6	&	4.8	&	5.1	&	0.01352	&	0.00396	&	0.23394	&	0.00751	\\
 \cmidrule{1-9}
500 &	1.3	&	2	&	1.5	&	1.8	&	0.03154	&	0.02332	&	0.05049	&	0.01887	\\
 &	3	&	1.5	&	2.8	&	2.5	&	0.04441	&	0.00052	&	0.04762	&	0.00383	\\
 &	2	&	2.2	&	6.5	&	4.1	&	0.01723	&	0.00406	&	0.19449	&	0.00315	\\
 &	1.4	&	1.6	&	4.8	&	5.1	&	0.00633	&	0.00143	&	0.07816	&	0.00286	\\
\noalign{\smallskip}\hline
\end{tabular}
\end{table}  

\section{Results and Discussion}
The hourly wind speed data measured at 10 m above ground level were collected by the National Institute of Meteorology of Brazil (INMET). The anemometers used for measuring the wind speed (in m/s) are installed in stations located in five cities spread in four states of the Brazilian Northeastern Region, as illustrated in Figure~\ref{fig:nordeste} (blue dots indicate the geographical position of the stations). Esperantina (denoted by Station 1) is located in the north part of the State of Piau\'i. Jaguaruana (denoted by Station 2) is located in the mesoregion of Jaguaribe in the State of Cear\'a. Cabaceiras (denoted by Station 3) and Monteiro (denoted by Station 4) are located in the mesoregion of Borborema, State of Para\'iba. Arapiraca (denoted by Station 5) is located in the mesoregion of Agreste, State of Alagoas. Table~\ref{tab:ngg_description} brings further details about the stations and the years of wind data available.

\begin{figure}
 \centering
 \includegraphics[scale=0.55]{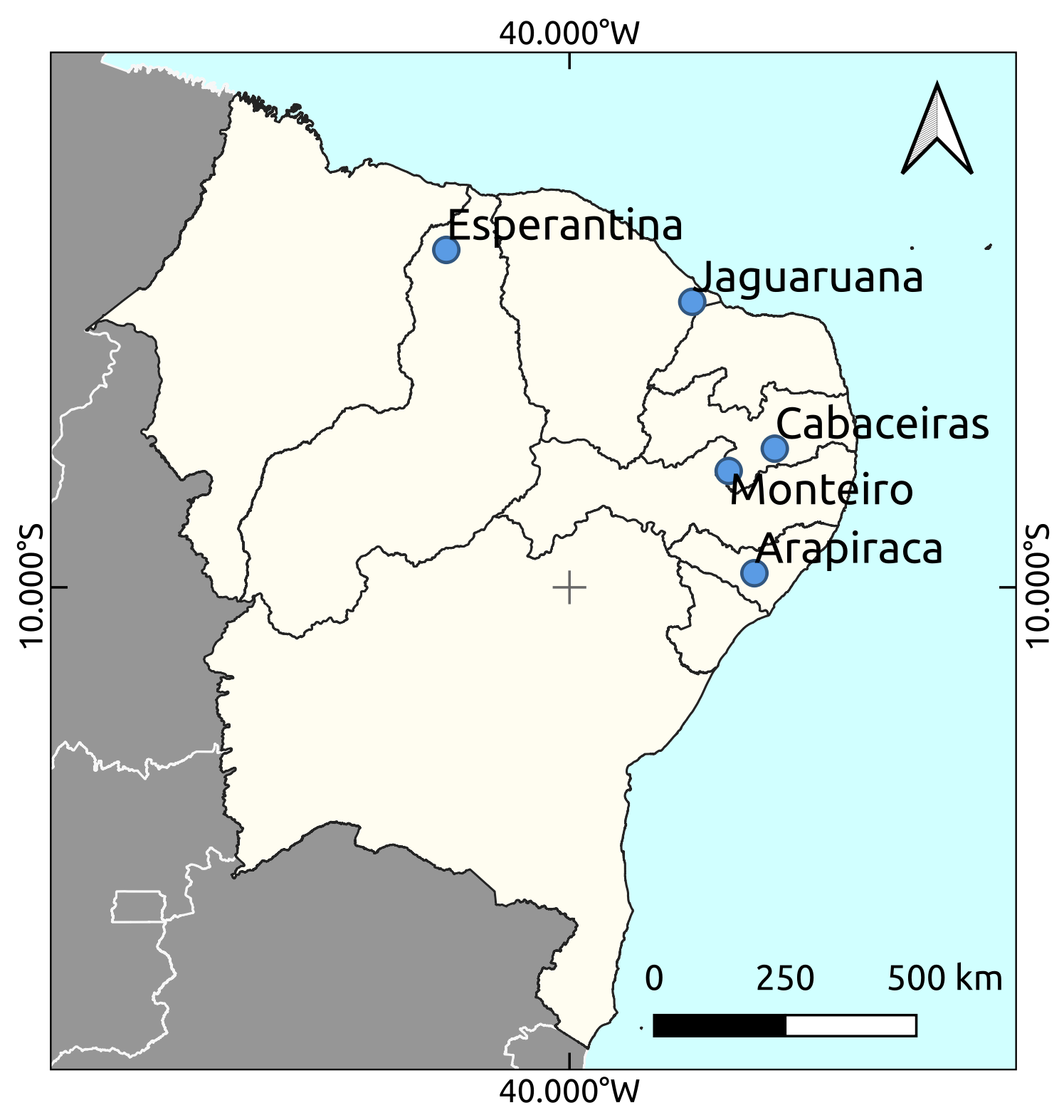}
 \caption{Northeastern Region of Brazil and geographical position of the stations.}
 \label{fig:nordeste}
\end{figure}

\begin{table}[h]
\caption{Details of the regions where the wind speed was measured.}
\label{tab:ngg_description} 
\centering
\begin{tabular}{llcccc}
\hline\noalign{\smallskip}
& Station & Latitude & Longitude & Altitude (m) & Period  \\
\noalign{\smallskip}\hline\noalign{\smallskip}
1 & Esperantina & $3^{\circ} 54' 07''$S & $42^{\circ} 14' 02''$W & 59 & 2007--2018 \\
2 & Jaguaruana & $4^{\circ} 50' 02''$S & $37^{\circ} 46' 51''$W & 20 & 2007--2018\\
3 & Cabaceiras & $7^{\circ} 29' 20''$S & $36^{\circ} 17' 13''$W & 382 & 2008--2018\\
4 & Monteiro & $7^{\circ} 53' 20''$S & $37^{\circ} 07' 12''$W & 599 & 2007--2018\\
5 & Arapiraca & $9^{\circ} 45' 07''$S & $36^{\circ} 39' 39''$W & 264 & 2008--2018\\
\noalign{\smallskip}\hline
\end{tabular}
\end{table}

As we can see in Table~\ref{tab:ngg_descriptive}, Station 4 has the highest value of the mean among the stations in the study, whereas the highest value of variance belongs to Station 2. Except for Station 3, whose skewness is negative, the remaining stations have different degrees of positive skewness. On the other hand, Station 1 has the only positive value of kurtosis and Station 5 has the lowest one. Thus, the descriptive statistics indicate that the statistical characteristics of the wind speed differ from station to station.     

\begin{table}[h!]
\caption{Descriptive statistics.}
\label{tab:ngg_descriptive} 
\centering
\begin{tabular}{lccccccrr}
\hline\noalign{\smallskip}
St. & $n$ & mean & median & min & max & variance & skewness & kurtosis \\
\noalign{\smallskip}\hline\noalign{\smallskip}
1 & 72637 & 1.56472 & 1.5 & 0.1 & 8.8 & 0.96082 & 0.74469 & 0.61585 \\
2 & 71797 & 3.09401 & 3 & 0.1 & 9.4 & 2.94352 & 0.25522 & $-0.59481$\\
3 & 83953 & 3.21937 & 3.3 & 0.1 & 9.9 & 2.49808 & $-0.08484$ & $-0.70967$\\
4 & 76738 & 3.28302 & 3.3 & 0.1 & 9.6 & 2.60513 & 0.16094 & $-0.50692$\\
5 & 72675 & 2.87482 & 2.8 & 0.1 & 9.3 & 2.80844 & 0.17049 & $-0.92422$\\
\noalign{\smallskip}\hline
\end{tabular}
\end{table}

We calculated the estimates of the parameters under the method of maximum likelihood for five distributions. Besides fitting the proposed model~(\ref{sec:model}), we fitted the Normal-Normal mixture model (NN), the Weibull-Weibull mixture model (WW), the Normal distribution (N) and the Weibull distribution (W) to each one of the five datasets. Table~\ref{tab:ngg_par1} presents the MLEs along with the respective standard errors in parentheses. The global maximum of the log-likelihood function was found using the L-BFGS-B algorithm~\citep{Byrd1995} for the distribution NWW, whereas the optimization concerning the mixture models NN and WW was performed along the lines of the EM-algorithm presented in \citet{Nguyen2018}. The standard errors are small in all scenarios, suggesting that the estimates in Table~\ref{tab:ngg_par1} are fairly accurate for the five distributions.     

\begin{table}[h!]
\caption{Estimates and standard errors in parentheses.}
\label{tab:ngg_par1} 
\centering
\begin{tabular}{llrrrrrr}
\hline\noalign{\smallskip}
Distr. & Par. & St1 & St2 & St3 & St4 & St5 \\
\noalign{\smallskip}\hline\noalign{\smallskip}
NWW	&	$k_1$	&	1.09455	  &  0.97219	&	0.93287	    &	1.16262	    &	0.92095	\\
	&		    &	(0.0066)  &  (0.0047)	&	(0.0035)    &	(0.0051)	&	(0.0041)	\\
 & $\lambda_1$  &	2.44439	  & 4.14273	    &	5.11736	    &	4.61159	    &	3.26931	\\
	&		    &	(0.0116)  & (0.0300)	&	(0.0151)	&	(0.0358)	&	(0.0157)	\\
	&	$k_2$	&	1.24356   & 2.21499	    &	2.74024	    &	2.17717	    &	2.87659	\\
	&		    &	(0.0090)  & (0.0181)	&	(0.0141)	&	(0.0232)	&	(0.0135)	\\
 & $\lambda_2$	&	2.23302   & 4.82613	    &	4.47508	    &	4.74835	    &	4.76046	\\
	&		    &	(0.0104)  & (0.0129)	&	(0.0083)	&	(0.0187)	&	(0.0059)	\\
NN	&	$\mu_1$	&	0.99225	& 1.25346	&	1.05246	&	1.52464	&	1.09819	\\
	&		&	(0.0088)	& (0.0112)	&	(0.0135)	&	(0.0152)	&	(0.0091)	\\
	&	$\sigma_1$	&	0.57673	&	0.73618	&	0.64260	&	0.80711	&	0.66916	\\
	&		&	(0.0051)	&  (0.0067)	&	(0.0077)	&	(0.0082)	&	(0.0056)	\\
	&	$\mu_2$	&	2.15056	&  3.69950	&	3.71423	&	3.78985	&	3.69032	\\
	&		&	(0.0121)	&  (0.0110)	&	(0.0094)	&	(0.0127)	&	(0.0118)	\\
	&	$\sigma_2$	&	0.96181	&  1.50081	&	1.28725	&	1.42131	&	1.33344	\\
	&		&	(0.0038)	&  (0.0051)	&	(0.0054)	&	(0.0056)	&	(0.0063)	\\
	&	$w$	&	0.50577	&  0.24754	&	0.18591	&	0.22375	&	0.31460	\\
	&		&	(0.0083)	&  (0.0041)	&	(0.0036)	&	(0.0054)	&	(0.0041)	\\
WW	&	$k_1$	&	2.46232	&  2.61226	&	3.48661	&	2.71890	&	3.63802	\\
	&		&	(0.0451)	&  (0.0162)	&	(0.0183)	&	(0.0146)	&	(0.0284)	\\
	&	$\lambda_1$	&	2.10941	&  4.09404	&	4.26754	&	4.08741	&	4.41073	\\
	&		&	(0.0089)	&  (0.0126)	&	(0.0070)	&	(0.0103)	&	(0.0110)	\\
	&	$k_2$	&	1.27757	&  1.20035	&	1.30110	&	1.26447	&	1.36100	\\
	&		&	(0.0106)	&  (0.0088)	&	(0.0069)	&	(0.0118)	&	(0.0056)	\\
	&	$\lambda_2$	&	1.42208	&  1.54955	&	1.81537	&	1.65794	&	1.86015	\\
	&		&	(0.0186)	&  (0.0392)	&	(0.0242)	&	(0.0526)	&	(0.0183)	\\
	&	$w$	&	0.44618	&  0.75154	&	0.71389	&	0.83225	&	0.51531	\\
	&		&	(0.0183)	&  (0.0075)	&	(0.0046)	&	(0.0069)	&	(0.0056)	\\
N	&	$\mu$	&	1.56471	&  3.09401	&	3.21937	&	3.28301	&	2.87482	\\
	&		&	(0.0036)	&  (0.0064)	&	(0.0054)	&	(0.0058)	&	(0.0062)	\\
	&	$\sigma$	&	0.98021	&  1.71565	&	1.58052	&	1.61403	&	1.67582	\\
	&		&	(0.0025)	& 	(0.0045)	&	(0.0038)	&	(0.0041)	&	(0.0043)	\\
W	&	$k$	&	1.58959	&	1.77193	&	2.03559	&	2.07474	&	1.65091	\\
	&		&	(0.0047)	& (0.0054)	&	(0.0059)	&	(0.0061)	&	(0.0051)	\\
	&	$\lambda$	&	1.73825	& 3.44821	&	3.59712	&	3.68213	&	3.18827	\\
	&		&	(0.0042) &	(0.0075)	&	(0.0063)	&	(0.0066)	&	(0.0074)	\\
\noalign{\smallskip}\hline
\end{tabular}
\end{table}

Four information criteria were used to perform comparisons among the fitted models. Generally, such criteria indicate that the best model is the one presenting the lowest value, since they are related to the amount of information lost by a given model. We used the well-known Akaike information criterion (AIC), consistent Akaike information criterion (CAIC), Bayesian information criterion (BIC) and Hannan-Quinn information criterion (HQIC). 
The statistics of Anderson-Darling (A$^*$) and Cram\'er-von Mises (W$^*$)~\citep{Chen1995} were also used to compare the fitted models. Since these statistics are measures of the difference between the empirical distribution function and the real underlying cdf, it is reasonable to say that the smaller their values, the better the fit. Table~\ref{tab:ngg_gof01} brings the aforementioned goodness-of-fit measures for the five cited models fitted to each station. 

\begin{table}[h!]
\caption{Goodness-of-fit measures.}
\label{tab:ngg_gof01} 
\centering
\begin{tabular}{llrrrrrr}
\hline\noalign{\smallskip}
Crit. & Distr. & St1 & St2 & St3 & St4 & St5 \\
\noalign{\smallskip}\hline\noalign{\smallskip}
AIC	&	NWW	& 189324.5 & 272976.1 & 306366.6 & 286607.1 & 268213.1 \\
	&	NN	&	197221.2	&	277168.5	&	310138.0	&	289011.9	&	273287.3 \\
	&	WW	&	189379.2	&	273075.2	&	306370.5	&	286647.3	&	268244.1 \\
	&	N	&	203235.1	&	281266.5	&	315112.0	&	291251.2	&	281292.0 \\
	&	W	&	190666.7	&	278737.7	&	319347.3	&	291232.9	&	277572.1 \\
CAIC	&	NWW	&	189365.2	&	273016.8	&	306408.0	&	286648.1	&	268253.8	\\
	&	NN	&	197272.2	&	277219.4	&	310189.7	&	289063.1	&	273338.3	\\
	&	WW	&	189430.2	&	273126.1	&	306422.2	&	286698.5	&	268295.1	\\
	&	N	&	203255.4	&	281286.9	&	315132.7	&	291271.7	&	281312.4	\\
	&	W	&	190687.1	&	278758.0	&	319368.0	&	291253.4	&	277592.5	\\
BIC	&	NWW	& 189361.2 & 273012.8 & 306404.0 & 286644.1 & 268249.8 \\
	&	NN	&	197267.2	&	277214.4	&	310184.7	&	289058.1	&	273333.3 \\
	&	WW	&	189425.2	&	273121.1	&	306417.2	&	286693.5	&	268290.1 \\
	&	N	&	203253.4	&	281284.9	&	315130.7	&	291269.7	&	281310.4 \\
	&	W	&	190685.1	&	278756.0	&	319366.0	&	291251.4	&	277590.5 \\
HQIC & NWW	& 189335.8 & 272987.4 & 306378.0 & 286618.5 & 268224.4 \\
	&	NN	&	197235.3	&	277182.7	&	310152.3	&	289026.1	&	273301.4 \\
	&	WW	&	189393.4	&	273089.3	&	306384.8	&	286661.5	&	268258.3 \\
	&	N	&	203240.7	&	281272.2	&	315117.7	&	291256.9	&	281297.7 \\
	&	W	&	190672.4	&	278743.3	&	319353.0	&	291238.6	&	277577.7 \\
A*	&	NWW	& 71.78 & 22.43 & 22.60 & 19.31 & 30.98 \\
	&	NN	&	181.13	&	51.72	&	54.33	&	27.79	&	78.72 \\
	&	WW	&	81.27	&	31.05	&	32.44	&	24.62	&	50.25 \\
	&	N	&	531.70	&	230.18	&	303.82	&	118.73	&	494.43 \\
	&	W	&	257.57	&	471.01	&	1117.26	&	337.07	&	784.70 \\
W*	&	NWW	& 7.73 & 2.33 & 2.94 & 2.58 & 3.36 \\
	&	NN	&	19.49	&	5.17	&	6.60	&	3.21	&	7.61 \\
	&	WW	&	9.35	&	3.59	&	4.22	&	3.89	&	4.82 \\
	&	N	&	72.08	&	31.04	&	45.14	&	16.72	&	72.83 \\
	&	W	&	33.87	&	62.25	&	164.54	&	46.18	&	108.58 \\
\noalign{\smallskip}\hline
\end{tabular}
\end{table}

According to Table~\ref{tab:ngg_gof01}, the distributions NWW and WW present the better fits among the competing models for all stations. The four information criteria indicate that NWW presents a higher performance over WW concerning stations 1, 2, 4 and 5. Regarding station 3, the difference between the values of AIC of NWW and WW is not considerable, although the AIC for NWW is slightly smaller in such case; the same behavior states for CAIC, BIC and HQIC. 

Both goodness-of-fit statistics A$^*$ and W$^*$ (see last ten rows of Table~\ref{tab:ngg_gof01}) agree with the information criteria in pointing NWW and WW as the two better fits. However, they suggest that NWW outperforms WW in fitting the datasets for all the five stations. To get more insight into these results, plots of the fitted densities overlapping the histograms of the wind speed data for stations 1 to 5 and the corresponding cdfs are presented in Figure~\ref{fig:nggapp}. 

\begin{figure}[h!]
 \centering
 \includegraphics[scale=0.8]{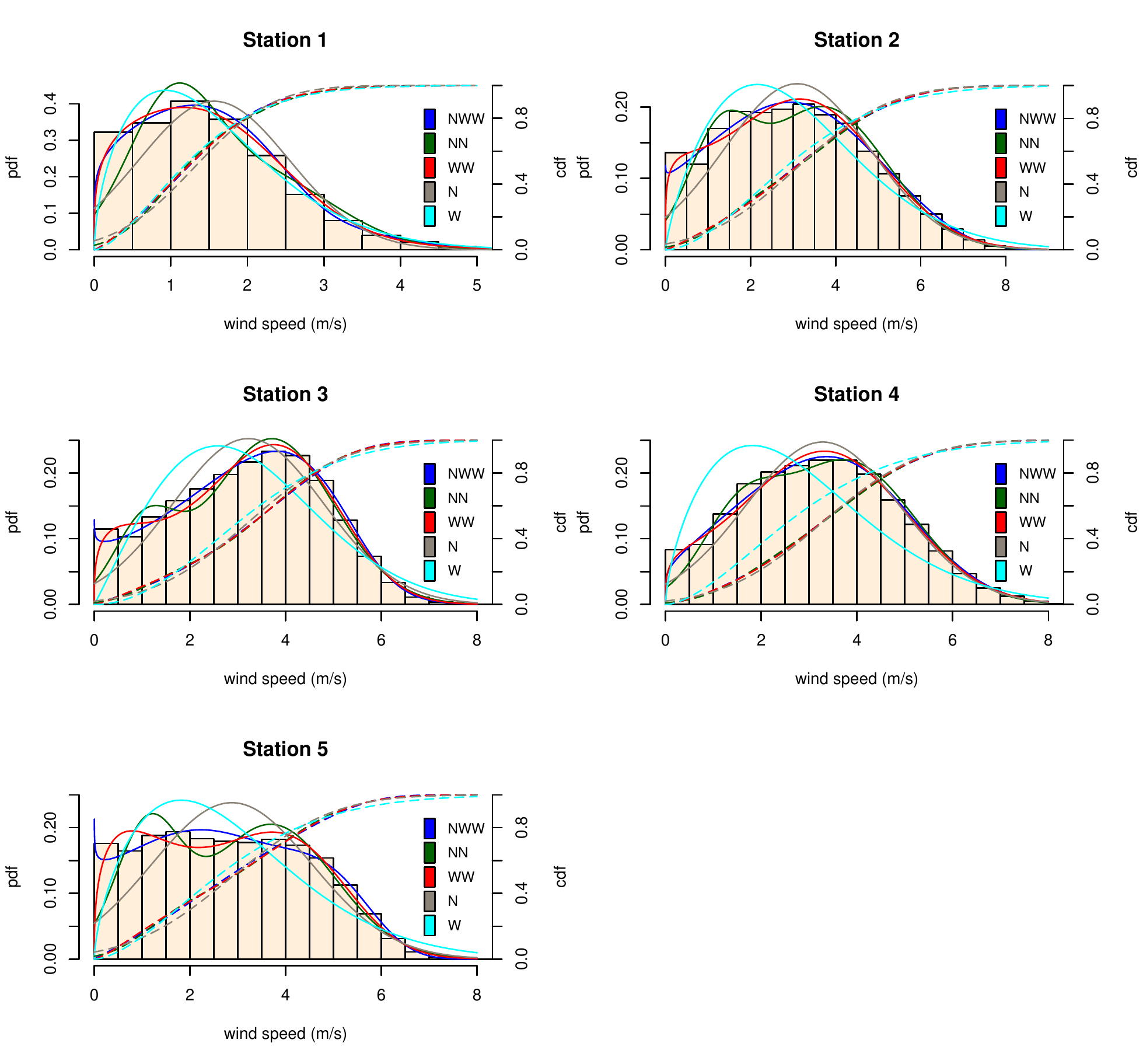}
 \caption{Histograms and fitted densities.}
 \label{fig:nggapp}
\end{figure}

It is worth pointing out that mixture models are commonly used to fit non-unimodal datasets, such as those represented by the histograms of stations 2, 3 and 5. Nonetheless, the results attest that the NWW accommodates such data better than the two mixture models presented in this study. Furthermore, NWW has one parameter less than NN or WW do.  

Finally, since the NWW distribution outperforms the competing models in fitting the wind speed data of the Northeastern Region of Brazil, according to different information criteria and formal goodness-of-fit statistics, we have good reasons to recommend its use to model similar data in future works. We also encourage practitioners of statistics to investigate the modelling benefits of the NWW (and other submodels from the Normal-($G_1$,$G_2$) class) with respect to data describing different phenomena usually modelled by mixtures.      

\section{Conclusions}
An alternative distribution for modelling wind speed data is proposed and some mathematical properties of the class that generates it are discussed, like the series representation of the pdf, the moments and the moment generating function. The general cdf of the Normal-$(G_1,G_2)$ class is written as a composition of two baselines and its submodels are identifiable as long as both baseline cdfs are. Such is the case of the NWW distribution.

The novel model has four parameters and high flexibility. It is able to fit right-skewed and left-skewed data and its pdf presents unimodal and bimodal shapes.

The Monte Carlo simulation studies indicate that the MLEs of the NWW parameters behave appropriately when the optimization is performed via the L-BFGS-B algorithm.

The modelling gains of the NWW distribution are upheld by the satisfactory results concerning the application to the wind speed data collected in the Northeastern Region of Brazil. The considered information criteria and the formal goodness-of-fit statistics of Anderson-Darling and Cram\'er-von Mises suggest that the proposed model outperforms other competing distributions that are commonly employed in wind speed modelling, especially the highly competitive mixture model of two Weibull components.  

We hope that this work may encourage the investigation of the modelling benefits of the NWW (and other identifiable submodels from the Normal-$(G_1,G_2)$ class) with respect to data describing other natural phenomena usually modelled by mixtures.

\appendix 
\section{} \label{app:a} 
The information matrix cited in section \ref{sec:ngginference} is given by $J(\bm{\theta}|\bm{X}) = -\nabla_{\bm{\theta}} {\nabla_{\bm{\theta}}}^{\top} \ell(\bm{\theta}|\bm{X}) = -(u_{jk})_{1\leq j \leq r+m,1\leq k \leq r+m}$ where:
\begin{align}
 u_{jk} = & \sum_{i=1}^n \frac{1}{f_{G_1,G_2}(x_i)} \phi \left(\frac{G_1(x_i)}{1-G_1(x_i)} \right) \frac{1}{(1-G_1(x_i))^2} \left\{ \left( \frac{2}{1-G_1(x_i)} - \frac{G_1(x_i)}{[1-G_1(x_i)]^3}\right) \right. \nonumber \\ 
\times & \left( \frac{\partial}{\partial \theta_k} G_1(x_i) \frac{\partial}{\partial \theta_j} g_1(x_i) + \frac{\partial}{\partial \theta_j} G_1(x_i) \frac{\partial}{\partial \theta_k} g_1(x_i) \right) + \frac{\partial^2}{\partial \theta_j \partial \theta_k} g_1(x_i) - \frac{g_1(x_i)}{(1-G_1(x_i))^2} \nonumber \\
\times & \left[ \frac{\partial}{\partial \theta_j} G_1(x_i) \frac{\partial}{\partial \theta_k} G_1(x_i) \left( \frac{3 {G_1}^2(x_i)}{(1-G_1(x_i))^2} + \frac{4 G_1(x_i)}{1-G_1(x_i)} +1\right) + \frac{G_1(x_i)}{1-G_1(x_i)} \right. \nonumber \\ \times & \left. \frac{\partial^2}{\partial \theta_j \partial \theta_k} G_1(x_i) \right] + \frac{g_1(x_i) G_1(x_i)}{(1-G_1(x_i))^5} \left( \frac{\partial}{\partial \theta_j} G_1(x_i) + \frac{\partial}{\partial \theta_k} G_1(x_i) \right) \nonumber \\
+& \frac{2 g_1(x_i)}{(1-G_1(x_i))^2} \left(3-\frac{2G_1(x_i)}{(1-G_1(x_i))^2} \right) \frac{\partial}{\partial \theta_k} G_1(x_i) \frac{\partial}{\partial \theta_j} G_1(x_i) + \frac{\partial^2}{\partial \theta_j \partial \theta_k} G_1(x_i) \nonumber \\
\times & \left. \frac{2g_1(x_i)}{1-G_1(x_i)} \right\} - \sum_{i=1}^n \frac{1}{{f}^2_{G_1,G_2}(x_i)} \phi^2 \left(\frac{G_1(x_i)}{1-G_1(x_i)} \right) \frac{1}{(1-G_1(x_i))^4} \nonumber \\
\times & \left[ \frac{\partial}{\partial \theta_j} g_1(x_i) + \left( \frac{2 g_1(x_i)}{1-G_1(x_i)} - \frac{g_1(x_i) G_1(x_i)}{(1-G_1(x_i))^3} \right) \frac{\partial}{\partial \theta_j} G_1(x_i) \right] \nonumber \\
\times & \left[ \frac{\partial}{\partial \theta_k} g_1(x_i) + \left( \frac{2 g_1(x_i)}{1-G_1(x_i)} - \frac{g_1(x_i) G_1(x_i)}{(1-G_1(x_i))^3} \right) \frac{\partial}{\partial \theta_k} G_1(x_i) \right], \mathrm{for}\; 1 \leq j \leq r,1 \leq k \leq r; \nonumber
\end{align}

\begin{align}
 u_{jk} = & \sum_{i=1}^n \frac{-1}{{f}^2_{G_1,G_2}(x_i)} \phi \left(\frac{G_1(x_i)}{1-G_1(x_i)} \right) \frac{1}{(1-G_1(x_i))^2} \frac{\phi(\log[1-G_2(x_i)])}{1-G_2(x_i)} \left( \frac{\partial}{\partial \theta_k} g_1(x_i) \right. \nonumber \\
 + & \left. \left[ \frac{2g_1(x_i)}{1-G_1(x_i)} - \frac{g_1(x_i) G_1(x_i)}{(1-G_1(x_i))^3} \right] \frac{\partial}{\partial \theta_k} G_1(x_i) \right) \left( \frac{\partial}{\partial \theta_j} g_2(x_i) + (1+\log[1-G_2(x_i)]) \right. \nonumber \\
 \times & \left. \frac{g_2(x_i)}{1-G_2(x_i)} \frac{\partial}{\partial \theta_j} G_2(x_i) \right), \mathrm{for}\; r < j \leq r+m,1 \leq k \leq r; \nonumber
\end{align}

\begin{align}
 u_{jk} = & \sum_{i=1}^n \frac{-1}{{f}^2_{G_1,G_2}(x_i)} \phi \left(\frac{G_1(x_i)}{1-G_1(x_i)} \right) \frac{1}{(1-G_1(x_i))^2} \frac{\phi(\log[1-G_2(x_i)])}{1-G_2(x_i)} \left( \frac{\partial}{\partial \theta_j} g_1(x_i) \right. \nonumber \\
 + & \left. \left[ \frac{2g_1(x_i)}{1-G_1(x_i)} - \frac{g_1(x_i) G_1(x_i)}{(1-G_1(x_i))^3} \right] \frac{\partial}{\partial \theta_j} G_1(x_i) \right) \left( \frac{\partial}{\partial \theta_k} g_2(x_i) + (1+\log[1-G_2(x_i)]) \right. \nonumber \\
 \times & \left. \frac{g_2(x_i)}{1-G_2(x_i)} \frac{\partial}{\partial \theta_k} G_2(x_i) \right), \mathrm{for}\; 1 \leq j \leq r, r < k \leq r+m; \nonumber
\end{align}

\begin{align}
 u_{jk} = & \sum_{i=1}^n \frac{1}{f_{G_1,G_2}(x_i)} \frac{\phi(\log[1-G_2(x_i)])}{1-G_2(x_i)} \left[ \frac{1+\log[1-G_2(x_i)]}{1-G_2(x_i)} \left( \frac{\partial}{\partial \theta_j} G_2(x_i) \frac{\partial}{\partial \theta_k} g_2(x_i) \right. \right. \nonumber \\
 + & \left. \frac{\partial}{\partial \theta_k} G_2(x_i) \frac{\partial}{\partial \theta_j} g_2(x_i) \right) + \frac{\partial^2}{\partial \theta_j \partial \theta_k} g_2(x_i) + \frac{\partial}{\partial \theta_j} G_2(x_i) \frac{\partial}{\partial \theta_k} G_2(x_i) \frac{g_2(x_i)}{(1-G_2(x_2))^2} \nonumber \\
 \times & \left(1+ \log^2[1-G_2(x_i)] + 3 \log[1-G_2(x_i)] \right) + (1+\log[1-G_2(x_i)]) \frac{\partial^2}{\partial \theta_j \partial \theta_k} G_2(x_i)  \nonumber \\
 \times & \left. \frac{g_2(x_i)}{1-G_2(x_i)} \right] - \sum_{i=1}^n \frac{1}{{f}^2_{G_1,G_2}(x_i)} \frac{\phi^2 (\log[1-G_2(x_i)])}{(1-G_2(x_i))^2} \left[ \frac{(1+\log[1-G_2(x_i)])g_2(x_i)}{1-G_2(x_i)} \right. \nonumber \\
 \times & \left. \frac{\partial}{\partial \theta_j} G_2(x_i) +\frac{\partial}{\partial \theta_j} g_2(x_i) \right] \left[ \frac{(1+\log[1-G_2(x_i)])g_2(x_i)}{1-G_2(x_i)} \frac{\partial}{\partial \theta_k} G_2(x_i) +\frac{\partial}{\partial \theta_k} g_2(x_i) \right], \nonumber \\
& \mathrm{for}\; r < j \leq r+m, r < k \leq r+m. \nonumber 
\end{align}


%
%


\bibliographystyle{elsarticle-num-names}
\bibliography{nww_preprint}   

\begin{thebibliography}{28}
\expandafter\ifx\csname natexlab\endcsname\relax\def\natexlab#1{#1}\fi
\providecommand{\url}[1]{\texttt{#1}}
\providecommand{\href}[2]{#2}
\providecommand{\path}[1]{#1}
\providecommand{\DOIprefix}{doi:}
\providecommand{\ArXivprefix}{arXiv:}
\providecommand{\URLprefix}{URL: }
\providecommand{\Pubmedprefix}{pmid:}
\providecommand{\doi}[1]{\href{http://dx.doi.org/#1}{\path{#1}}}
\providecommand{\Pubmed}[1]{\href{pmid:#1}{\path{#1}}}
\providecommand{\bibinfo}[2]{#2}
\ifx\xfnm\relax \def\xfnm[#1]{\unskip,\space#1}\fi
\bibitem[{Ara{\'u}jo et~al.(2020)Ara{\'u}jo, Souza, Meireles, and
  Brannstrom}]{Araujo2020}
\bibinfo{author}{J.~C.~H. Ara{\'u}jo}, \bibinfo{author}{W.~F.~d. Souza},
  \bibinfo{author}{A.~J. d.~A. Meireles}, \bibinfo{author}{C.~Brannstrom},
\newblock \bibinfo{title}{Sustainability challenges of wind power deployment in
  coastal {C}ear{\'a} state, {B}razil},
\newblock \bibinfo{journal}{Sustainability} \bibinfo{volume}{12}
  (\bibinfo{year}{2020}) \bibinfo{pages}{5562}.
\bibitem[{Perkin et~al.(2015)Perkin, Garrett, and Jensson}]{Perkin2015}
\bibinfo{author}{S.~Perkin}, \bibinfo{author}{D.~Garrett},
  \bibinfo{author}{P.~Jensson},
\newblock \bibinfo{title}{Optimal wind turbine selection methodology: A
  case-study for {B}\'{u}rfell, {I}celand},
\newblock \bibinfo{journal}{Renewable Energy} \bibinfo{volume}{75}
  (\bibinfo{year}{2015}) \bibinfo{pages}{165–172}.
\bibitem[{Eltamaly(2013)}]{Eltamaly2013}
\bibinfo{author}{A.~Eltamaly},
\newblock \bibinfo{title}{Design and implementation of wind energy system in
  {S}audi {A}rabia},
\newblock \bibinfo{journal}{Renewable Energy} \bibinfo{volume}{60}
  (\bibinfo{year}{2013}) \bibinfo{pages}{42--52}.
\bibitem[{Ilhan and Kantar(2012)}]{Ilhan2012}
\bibinfo{author}{U.~Ilhan}, \bibinfo{author}{Y.~M. Kantar},
\newblock \bibinfo{title}{Analysis of some flexible families of distribution
  for estimation of wind speed distributions},
\newblock \bibinfo{journal}{Applied Energy} \bibinfo{volume}{89}
  (\bibinfo{year}{2012}) \bibinfo{pages}{355–367}.
\bibitem[{Pishgar-Komleh et~al.(2015)Pishgar-Komleh, Keyhani, and
  Sefeedpari}]{Pishgar2015}
\bibinfo{author}{S.~Pishgar-Komleh}, \bibinfo{author}{A.~Keyhani},
  \bibinfo{author}{P.~Sefeedpari},
\newblock \bibinfo{title}{Wind speed and power density analysis based on
  {W}eibull and {R}ayleigh distributions (a case study: {F}irouzkooh county of
  {I}ran)},
\newblock \bibinfo{journal}{Renewable and Sustainable Energy Reviews}
  \bibinfo{volume}{42} (\bibinfo{year}{2015}) \bibinfo{pages}{313 -- 322}.
\bibitem[{Safari(2011)}]{Safari2011}
\bibinfo{author}{B.~Safari},
\newblock \bibinfo{title}{Modeling wind speed and wind power distributions in
  {R}wanda},
\newblock \bibinfo{journal}{Renewable and Sustainable Energy Reviews}
  \bibinfo{volume}{15} (\bibinfo{year}{2011}) \bibinfo{pages}{925--935}.
\bibitem[{Weisser(2003)}]{Weisser2003}
\bibinfo{author}{D.~Weisser},
\newblock \bibinfo{title}{A wind energy analysis of {G}renada: an estimation
  using the ‘{W}eibull’ density function},
\newblock \bibinfo{journal}{Renewable Energy} \bibinfo{volume}{28}
  (\bibinfo{year}{2003}) \bibinfo{pages}{1803 -- 1812}.
\bibitem[{Kollu et~al.(2012)Kollu, Rayapudi, Narasimham, and
  Pakkurthi}]{Kollu2012}
\bibinfo{author}{R.~Kollu}, \bibinfo{author}{S.~Rayapudi},
  \bibinfo{author}{S.~Narasimham}, \bibinfo{author}{K.~Pakkurthi},
\newblock \bibinfo{title}{Mixture probability distribution functions to model
  wind speed distributions},
\newblock \bibinfo{journal}{International Journal of Energy and Environmental
  Engineering} \bibinfo{volume}{3} (\bibinfo{year}{2012}).
\bibitem[{Bali and Theodossiou(2008)}]{Bali2008}
\bibinfo{author}{T.~G. Bali}, \bibinfo{author}{P.~Theodossiou},
\newblock \bibinfo{title}{Risk measurement performance of alternative
  distribution functions},
\newblock \bibinfo{journal}{Journal of Risk and Insurance} \bibinfo{volume}{75}
  (\bibinfo{year}{2008}) \bibinfo{pages}{411--437}.
\bibitem[{Hansen(1994)}]{Hansen1994}
\bibinfo{author}{B.~E. Hansen},
\newblock \bibinfo{title}{Autoregressive conditional density estimation},
\newblock \bibinfo{journal}{International Economic Review} \bibinfo{volume}{35}
  (\bibinfo{year}{1994}) \bibinfo{pages}{705--730}.
\bibitem[{Morgan et~al.(2011)Morgan, Lackner, Vogel, and Baise}]{Morgan2011}
\bibinfo{author}{E.~C. Morgan}, \bibinfo{author}{M.~Lackner},
  \bibinfo{author}{R.~M. Vogel}, \bibinfo{author}{L.~G. Baise},
\newblock \bibinfo{title}{Probability distributions for offshore wind speeds},
\newblock \bibinfo{journal}{Energy Conversion and Management}
  \bibinfo{volume}{52} (\bibinfo{year}{2011}) \bibinfo{pages}{15 -- 26}.
\bibitem[{Mohammadi et~al.(2017)Mohammadi, Alavi, and Mcgowan}]{Mohammadi2017}
\bibinfo{author}{K.~Mohammadi}, \bibinfo{author}{O.~Alavi},
  \bibinfo{author}{J.~Mcgowan},
\newblock \bibinfo{title}{Use of {B}irnbaum-{S}aunders distribution for
  estimating wind speed and wind power probability distributions: A review},
\newblock \bibinfo{journal}{Energy Conversion and Management}
  \bibinfo{volume}{143} (\bibinfo{year}{2017}) \bibinfo{pages}{109–122}.
\bibitem[{Qin et~al.(2011)Qin, Li, and Xiong}]{Qin2011}
\bibinfo{author}{Z.~Qin}, \bibinfo{author}{W.~Li}, \bibinfo{author}{X.~Xiong},
\newblock \bibinfo{title}{Estimating wind speed probability distribution using
  kernel density method},
\newblock \bibinfo{journal}{Electric Power Systems Research}
  \bibinfo{volume}{81} (\bibinfo{year}{2011}) \bibinfo{pages}{2139--2146}.
\bibitem[{Hu et~al.(2016)Hu, Li, Yang, and Wang}]{Bo2016}
\bibinfo{author}{B.~Hu}, \bibinfo{author}{Y.~Li}, \bibinfo{author}{H.~Yang},
  \bibinfo{author}{H.~Wang},
\newblock \bibinfo{title}{Wind speed model based on kernel density estimation
  and its application in reliability assessment of generating systems},
\newblock \bibinfo{journal}{Journal of Modern Power Systems and Clean Energy}
  \bibinfo{volume}{5} (\bibinfo{year}{2016}).
\bibitem[{Han et~al.(2019)Han, Ma, Wang, and Chu}]{Han2019}
\bibinfo{author}{Q.~Han}, \bibinfo{author}{S.~Ma}, \bibinfo{author}{T.~Wang},
  \bibinfo{author}{F.~Chu},
\newblock \bibinfo{title}{Kernel density estimation model for wind speed
  probability distribution with applicability to wind energy assessment in
  {C}hina},
\newblock \bibinfo{journal}{Renewable and Sustainable Energy Reviews}
  \bibinfo{volume}{115} (\bibinfo{year}{2019}) \bibinfo{pages}{109387}.
\bibitem[{Jaramillo and Borja(2004)}]{Jaramillo2004}
\bibinfo{author}{O.~Jaramillo}, \bibinfo{author}{M.~Borja},
\newblock \bibinfo{title}{Wind speed analysis in {L}a {V}entosa, {M}exico: A
  bimodal probability distribution case},
\newblock \bibinfo{journal}{Renewable Energy} \bibinfo{volume}{29}
  (\bibinfo{year}{2004}) \bibinfo{pages}{1613--1630}.
\bibitem[{Chang(2011)}]{Chang2011}
\bibinfo{author}{T.~Chang},
\newblock \bibinfo{title}{Estimation of wind energy potential using different
  probability density functions},
\newblock \bibinfo{journal}{Applied Energy} \bibinfo{volume}{88}
  (\bibinfo{year}{2011}) \bibinfo{pages}{1848--1856}.
\bibitem[{Akdag et~al.(2010)Akdag, Bagiorgas, and Mihalakakou}]{Akdag2010}
\bibinfo{author}{S.~Akdag}, \bibinfo{author}{H.~Bagiorgas},
  \bibinfo{author}{G.~Mihalakakou},
\newblock \bibinfo{title}{Use of two-component {W}eibull mixtures in the
  analysis of wind speed in the {E}astern {M}editerranean},
\newblock \bibinfo{journal}{Applied Energy} \bibinfo{volume}{87}
  (\bibinfo{year}{2010}) \bibinfo{pages}{2566--2573}.
\bibitem[{Carta and Ram\'{i}rez(2007)}]{Carta2007}
\bibinfo{author}{J.~Carta}, \bibinfo{author}{P.~Ram\'{i}rez},
\newblock \bibinfo{title}{Analysis of two-component mixture {W}eibull
  statistics for estimation of wind speed distributions},
\newblock \bibinfo{journal}{Renewable Energy} \bibinfo{volume}{32}
  (\bibinfo{year}{2007}) \bibinfo{pages}{518--531}.
\bibitem[{McLachlan and Peel(2000)}]{mclachlan}
\bibinfo{author}{G.~McLachlan}, \bibinfo{author}{D.~Peel},
  \bibinfo{title}{Finite Mixture Models}, \bibinfo{publisher}{Wiley
  Interscience}, \bibinfo{year}{2000}.
\bibitem[{Brito et~al.(2019)Brito, Rego, Oliveira, and Gomes-Silva}]{Brito2019}
\bibinfo{author}{C.~R. Brito}, \bibinfo{author}{L.~C. Rego},
  \bibinfo{author}{W.~R. Oliveira}, \bibinfo{author}{F.~Gomes-Silva},
\newblock \bibinfo{title}{Method for generating distributions and classes of
  probability distributions: the univariate case},
\newblock \bibinfo{journal}{Hacettepe Journal of Mathematics and Statistics}
  \bibinfo{volume}{48} (\bibinfo{year}{2019}) \bibinfo{pages}{897--930}.
\bibitem[{Alzaatreh et~al.(2013)Alzaatreh, Lee, and Famoye}]{Alzaatreh2013}
\bibinfo{author}{A.~Alzaatreh}, \bibinfo{author}{C.~Lee},
  \bibinfo{author}{F.~Famoye},
\newblock \bibinfo{title}{A new method for generating families of continuous
  distributions},
\newblock \bibinfo{journal}{Metron} \bibinfo{volume}{71} (\bibinfo{year}{2013})
  \bibinfo{pages}{63--79}.
\bibitem[{Mudholkar and Srivastava(1993)}]{Mudholkar1993}
\bibinfo{author}{G.~S. Mudholkar}, \bibinfo{author}{D.~K. Srivastava},
\newblock \bibinfo{title}{Exponentiated {W}eibull family for analyzing bathtub
  failure-rate data},
\newblock \bibinfo{journal}{IEEE transactions on reliability}
  \bibinfo{volume}{42} (\bibinfo{year}{1993}) \bibinfo{pages}{299--302}.
\bibitem[{{R Core Team}(2018)}]{r2018}
\bibinfo{author}{{R Core Team}}, \bibinfo{title}{R: A Language and Environment
  for Statistical Computing}, \bibinfo{organization}{R Foundation for
  Statistical Computing}, \bibinfo{address}{Vienna, Austria},
  \bibinfo{year}{2018}. \URLprefix \url{http://www.R-project.org/}.
\bibitem[{{V}on {N}eumann(1951)}]{Neumann}
\bibinfo{author}{J.~{V}on {N}eumann}, \bibinfo{title}{Various techniques used
  in connection with random digits}, Applied Mathematics Series 12,
  \bibinfo{publisher}{National Bureau of Standards},
  \bibinfo{address}{Washington, DC, USA}, \bibinfo{year}{1951}.
\bibitem[{Byrd et~al.(1995)Byrd, Lu, Nocedal, and Zhu}]{Byrd1995}
\bibinfo{author}{R.~H. Byrd}, \bibinfo{author}{P.~Lu},
  \bibinfo{author}{J.~Nocedal}, \bibinfo{author}{C.~Zhu},
\newblock \bibinfo{title}{A limited memory algorithm for bound constrained
  optimization},
\newblock \bibinfo{journal}{SIAM Journal on Scientific Computing}
  \bibinfo{volume}{16} (\bibinfo{year}{1995}) \bibinfo{pages}{1190--1208}.
\bibitem[{Nguyen et~al.(2018)Nguyen, Wang, and McLachlan}]{Nguyen2018}
\bibinfo{author}{H.~D. Nguyen}, \bibinfo{author}{D.~Wang},
  \bibinfo{author}{G.~J. McLachlan},
\newblock \bibinfo{title}{Randomized mixture models for probability density
  approximation and estimation},
\newblock \bibinfo{journal}{Information Sciences} \bibinfo{volume}{467}
  (\bibinfo{year}{2018}) \bibinfo{pages}{135 -- 148}.
\bibitem[{Chen and Balakrishnan(1995)}]{Chen1995}
\bibinfo{author}{G.~Chen}, \bibinfo{author}{N.~Balakrishnan},
\newblock \bibinfo{title}{A general purpose approximate goodness-of-fit test},
\newblock \bibinfo{journal}{Journal of Quality Technology} \bibinfo{volume}{27}
  (\bibinfo{year}{1995}) \bibinfo{pages}{154--161}.

\end{thebibliography}







\end{document}